\documentclass[10pt,twocolumn]{article}

\makeatletter
\def\ps@headings{%
\def\@oddhead{\mbox{}\scriptsize\rightmark \hfil \thepage}%
\def\@evenhead{\scriptsize\thepage \hfil \leftmark\mbox{}}%
\def\@oddfoot{}%
\def\@evenfoot{}}
\makeatother
\usepackage{times}
\usepackage[cmex10]{amsmath}
\usepackage{amsfonts,amssymb,amsthm}
\usepackage{graphicx,color}
\usepackage{algorithmic}
\usepackage{algorithm}
\usepackage{latexsym}
\usepackage{cite}

\newcommand{\eat}[1]{}

\newtheorem{definition}{Definition}[section]
\newtheorem{corollary}{Corollary}[section]
\newtheorem{theorem}{Theorem}[section]
\newtheorem{lemma}{Lemma}[section]

\newtheorem{observation}{Observation}[section]

\newcommand{\set}[1]{\left\{ #1 \right\}}

\newcommand{\squeeze}[0]{\vspace{-1ex}}

\begin{document}

\title{Modeling and Analysis of Time-Varying Graphs\thanks{Research was sponsored by the Army Research Laboratory and was accomplished under Cooperative Agreement Number W911NF-09-2-0053. The views and conclusions contained in this document are those of the authors and should not be interpreted as representing the official policies, either expressed or implied, of the Army Research Laboratory or the U.S. Government. The U.S. Government is authorized to reproduce and distribute reprints for Government purposes notwithstanding any copyright notation here on.}}

\author{Prithwish Basu\\
Raytheon BBN Technologies\\
Cambridge, MA\\
pbasu@bbn.com\\
\and 
Amotz Bar-Noy\\
City University of New York\\
New York, NY\\
amotz@sci.brooklyn.cuny.edu\\
\and 
Ram Ramanathan\\
Raytheon BBN Technologies\\
Cambridge, MA\\
ramanath@bbn.com\\
\and 
Matthew P. Johnson\\
Pennsylvania State University\\
State College, PA\\
mpjohnson@gmail.com
}

\date{}
\maketitle


\begin{abstract}
We live in a world increasingly dominated by networks -- communications, social, information, biological etc.  A central attribute of many of these networks is that they are dynamic, that is, they exhibit structural changes over time. While the practice of dynamic networks has proliferated, we lag behind in the fundamental, mathematical understanding of network dynamism. Existing research on time-varying graphs ranges from preliminary algorithmic studies (e.g., Ferreira's work on evolving graphs) to analysis of specific properties such as flooding time in dynamic random graphs. A popular model for studying dynamic graphs is a sequence of graphs arranged by increasing snapshots of time. In this paper, we study the fundamental property of reachability in a time-varying graph over time and characterize the latency with respect to two metrics, namely store-or-advance latency and cut-through latency. Instead of expected value analysis, we concentrate on characterizing the exact probability distribution of routing latency along a randomly intermittent path in two popular dynamic random graph models. Using this analysis, we characterize the loss of accuracy (in a probabilistic setting) between multiple temporal graph models, ranging from one that preserves all the temporal ordering information for the purpose of computing temporal graph properties to one that collapses various snapshots into one graph (an operation called smashing), with multiple intermediate variants. We also show how some other traditional graph theoretic properties can be extended to the temporal domain. Finally, we propose algorithms for controlling the progress of a packet in single-copy adaptive routing schemes in various dynamic random graphs.
\end{abstract}


\section{Introduction}

We live in a world increasingly dominated by networks -- communications, social, biological etc -- imagine, for instance, an ad hoc infrastructureless communications network of constantly mobile
soldiers.  A central feature of many of these networks is that they are dynamic, that is, they exhibit structural changes over time. While the practice of dynamic networks has proliferated, especially in
the area of military communications networks, we lag behind in the fundamental, mathematical understanding of network dynamism.

Time-varying graphs have been a topic of active research recently~\cite{Ferreira04,Clementi08,Baumann09,Mucha10}. They are useful in the study of communication networks with intermittent connectivity such as delay-tolerant networks~\cite{Jain04} and even disruption-tolerant social networks~\cite{Hui05haggle}; duty cycling wireless sensor networks~\cite{Basu08,Chau09,Basu10}, and the like. Existing research on time-varying graphs ranges from algorithmic studies on {\em graph journeys} ~\cite{Ferreira04} to analysis of specific properties such as flooding time in dynamic random graphs~\cite{Clementi08,Baumann09}. Empirical simulation-based analysis of certain temporal graph properties such as temporal distance and temporal efficiency has also been a topic of recent research~\cite{Tang09}. 

In this paper, we propose a model of time-varying graphs called {\em Temporal Graphlets} which are essentially a time-series of static graph snapshots. While similar models have been studied in the literature before, albeit with alternative names such as space-time graphs~\cite{Merugu04}, we propose new research directions in temporal graph theory and present analytical results on two different aspects of this temporal graph model.

First, a directed {\em stacked} graph is created from all the temporal snapshots of the time-varying graph and we show how certain standard graph theoretic properties such as reachability, connectivity, etc. can be extended to this model. Then we propose a technique named {\em smashing} for collapsing all or parts of the temporal graph and analyze how the reachability property is affected due to the loss of temporal ordering information. We also introduce an intermediate model of {\em $m$-smashed} graphs which selectively collapse parts of the temporal graph while preserving the remaining stacked structure. We show how the degree of smashing can impact graph properties by means of a thorough comparative probabilistic analysis of the reachability property for the simple time-varying line network. This is potentially useful for online analysis of large temporal graphs where accuracy can be traded for speed and complexity.

We study two different metrics for measuring latency in this paper: (a) Store-or-advance; and (b) Cut-through. In the former, a message can be forwarded to only a neighbor in a unit time step, whereas in the latter, a message can be routed to any neighbor in the currently connected component {\em instantaneously}. In this paper, we study theoretical aspects of reachability in temporal graphs under various random edge-dynamics models. In particular, we characterize the exact probability distributions for latency (not just the first moment) and also a recursive form for message location in two popular dynamic random graph models for the dynamic line graph (or linear network topology), namely, the independent probabilistic model and the two-step Markov chain model.

Finally, we propose an adaptive routing algorithm that minimizes expected traversal time between a source and a destination node in the independent probabilistic temporal graph model.

This paper is organized as follows. Section \ref{sec:models} introduces deterministic and random models of temporal graphs. Section \ref{sec:latency} presents results on the probabilistic analysis of latency along dynamically changing random paths in graphs. Section \ref{sec:stgsmg} presents stacked and smashed graph models for temporal graphs and presents comparative probabilistic analysis of latency under both models for time-varying random paths. Section \ref{sec:routing} presents an adaptive routing algorithm in time-varying graphs. Section \ref{sec:discuss} concludes the paper with a discussion on future research directions.

\section{Models of Temporal Graphs} \label{sec:models}

Time-varying graphs occur commonly in the real world, and it is necessary to have mathematical models for their representation. We first introduce a deterministic model for representing a series of time-varying graphs, and propose two different models for routing in such graphs. We then propose enhancements to well known dynamic random graph models, which are used throughout this paper for analysis.

\subsection{Temporal Graphlets: A Deterministic Model of Dynamic Graphs}
\label{sec:tgmodel}

Assume slotted time starting at time $0$. Slot $t$ starts just after time $t-1$ and ends at time $t$. A Temporal Graphlet Sequence $TGS(T_1,T_2)=\{G(t)=(V(t),E(t))\},T_1\leq t \leq T_2$ is our basic deterministic model for a dynamic network and attempts to capture its space-time trajectory (see Figure \ref{fig:graphlets}). Each $G(t)$ is referred to as a Temporal Graphlet or simply Graphlet. Alternate notations that we will use, depending on the emphasis, include $G(T_1, T_2)$, $G[1,T]$ (shifting the frame of reference maintains properties), $G[T]$ (reference shifting is implied).

While traditional graph theory only considers properties in the ``horizontal'' (space) dimension, we consider properties across the ``vertical'' (time) dimension as well. For instance, $u \rightarrow v$
is $T$-reachable iff there exists a sequence of edges $(u_1, u_2), (u_2, u_3), ... (u_{m-1}, u_m)$, $u = u_1$, $v = u_m$ and $(u_i,u_{i+1})$ $\in$ $V(t_j)$, $1 \leq i < m$, $t_j \geq t_{j-1}$, $1 \leq t_j \leq T$. 
\eat{That is, there are graphlets (not necessarily consecutive) between 1 and $T$ such that there is a sequence of segments (a segment is an edge or a sequence of edges), one in each graphlet, with
the destination of a segment being the source of a segment in the temporally following graphlet.}

For example, in Figure \ref{fig:graphlets}, every graphlet is disconnected, but T-reachability holds for $a \rightarrow f$. Similarly, a T-cut is the removal of a set of vertices $X \subset V(1) \cup V(2) \cup V(3) \ldots \cup V(t)$ that results in some $u$ and $v$ losing their T-reachability property. Special or restricted temporal graphlets are also possible, e.g., a T-$k$-regular graph is one in which every node makes unique contact exactly $k$ times during its lifetime.

Assume a node $v$ wants to send a message to a certain node $u$. At the beginning of a slot the node that has the message can {\em store} it or {\em forward} it to another neighboring node. At the end of the slot the graph may change according to the TGS. There are two models for measuring progress accomplished by a message under the circumstances.

\begin{definition}
In the {\em Store or Advance} ({\sf SoA}) model, a node can forward the message only to one of its direct neighbors, and that is assumed to take a time slot. Even if the neighbor's neighboring edges are active right now, one may not be able to avail those edges right away. Instead, one has to wait for at least one (generally more) time slot(s) until the message reaches  the neighbor.
\end{definition}

\begin{definition}
In the {\em Cut-through} ({\sf CuT}) model, a node may send the message to any node in its connected component, and the entire connected component can be traversed instantaneously or at least in a much shorter time scale than that of edge dynamics.
\end{definition}

While the {\sf SoA} model finds more applications in most time-varying networks such as MANETs, DTNs, and social networks~\cite{Hui05haggle}, the {\sf CuT} model is interesting in its own right, and has been proposed in certain applications in low latency MANET design~\cite{Ram05}.

\begin{figure}[tbp]
\centering
\includegraphics[width=0.45\textwidth]{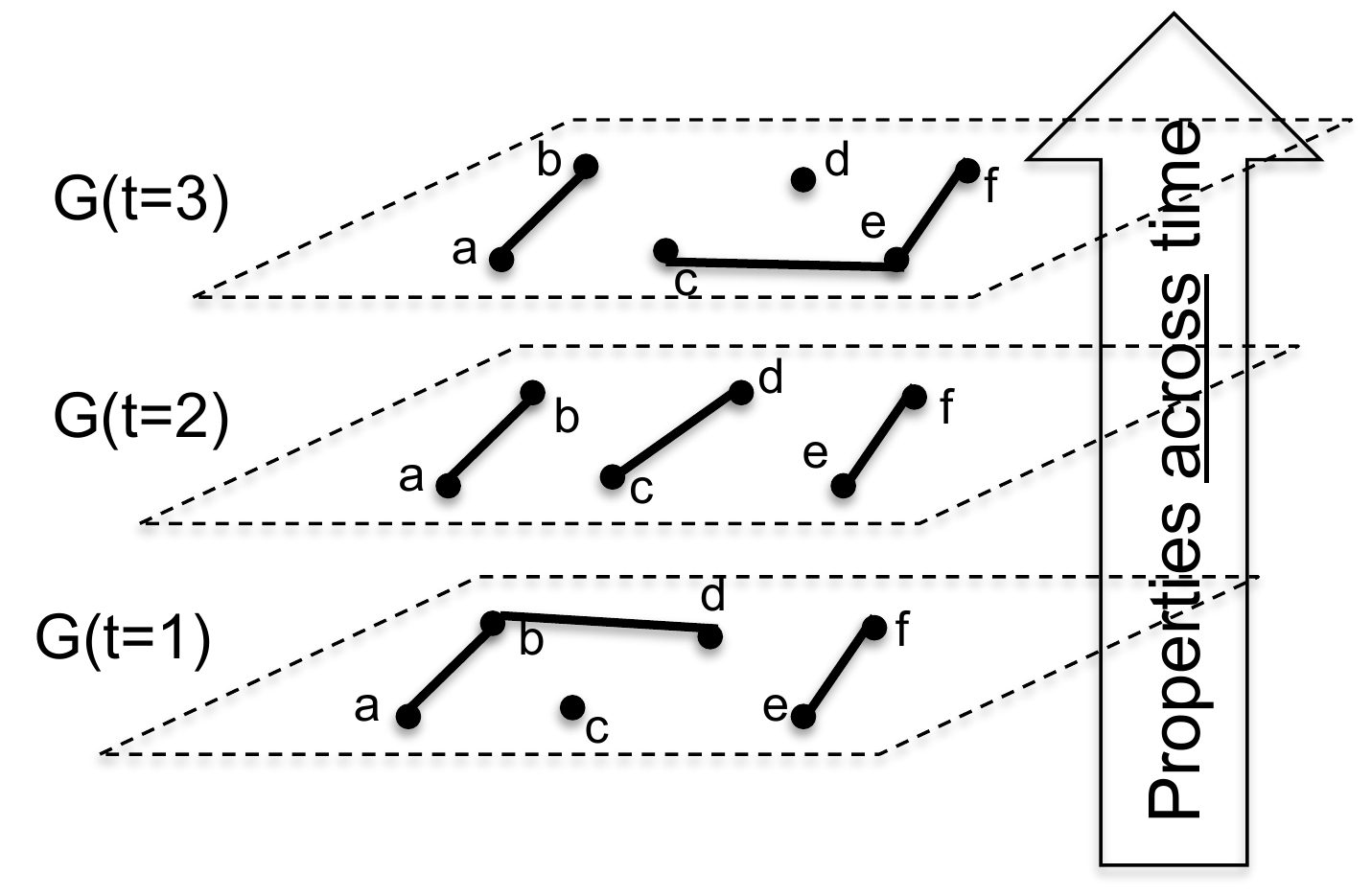}
\caption{Temporal Graphlets for $t$ = 1,2,3. $V(1) = V(2) = V(3)$ = \{a,b,c,d,e,f\}. Although this figure does not illustrate it, vertex set need not be the same (nodes could be added or deleted)}
\label{fig:graphlets}
\end{figure}

In Section \ref{sec:stgsmg}, we show how the deterministic temporal graphlet model can be useful for extending static graph theoretic properties to dynamic graphs. A related concept of {\em slices} has been proposed recently~\cite{Mucha10}. They define coupling variables between instances of the same node in consecutive slices. However, the focus of this work is on detecting communities over time.

\subsection{Stochastic Models of Dynamic Graphs}
\label{sec:stochmodel}

Random graph models are very useful for studying a plethora of graph properties in a probabilistic sense. A classic example of random graphs is the family of Erdos-Renyi graphs $ER(n,p)$ which are static graphs on $n$ nodes with any of the $n\choose 2$ edges existing with probability $p$. The probability of the existence of an edge is independent of that of another edge in the graph. Although too simplistic and perhaps unrealistic for many application scenarios, random graphs have played a big role in the development of a good understanding of key physical phenomena such as phase transitions and percolation~\cite{Grimmett99}.

Researchers have proposed adding a time dimension to the static random graph model such that time is slotted and each edge in the graph exists in each time slot with probability $p$ and does not exist with probability $1-p$~\cite{Clementi07}. We refer to this graph as the {\em dynamic} Erdos-Renyi graph. 

\begin{definition}
{\em Dynamic $ER(p)$ graphs}:  $G_t$ which is the graph at the end of slot $t$ and at the beginning of slot $t+1$ is drawn from the family of graphs $ER(n,p)$. $G_0$ is the initial graph and $G_T$ is the final graph if the time horizon ends at time $T$.
\end{definition}

\begin{definition}
{\em Markovian $(q,p)$ graphs}: In this model of dynamic random graphs~\cite{Clementi08}, each edge in $G_t$ can be in one of two states, ON or OFF, and the probability distribution is governed by a two-state Markov chain. The transition probabilities are given by $P(OFF\rightarrow ON)=p$, $P(OFF\rightarrow OFF)=1-p$, $P(ON\rightarrow OFF)=q$, and $P(ON\rightarrow ON)=1-q$.
\end{definition}

We propose a generic enhancement to these two dynamic random graph models. Instead of allowing a stochastic process to act on all of the possible $n\choose 2$ edges, we restrict it to act on only the edges in a given {\em underlying graph}, $G_u$. Clearly, when $G_u = K_n$, the complete graph, these stochastic process applies to all possible edges, and then this is equivalent to the older model. 
\eat{However, several interesting special cases for $G_u$ exist, e.g., line, cycle, star, tree, or even a random graph. Usually the underlying graph will have $n$ vertices denoted by $\set{1,\ldots,n}$ but sometimes it would be more convenient to have $n+1$ vertices denoted by $\set{0,...,n}$.}

\begin{observation}
The Markov $(1,1)$ dynamic graph corresponds to the family of perfectly alternating graphs, $(G_t,G_{t+1})$, such that $G_{t+1}$ has all the edges that do not exist in $G_t$, and vice versa.
\end{observation}

\begin{observation}
At any time slot, if $G_u=K_n$, the $(1-p,p)$ Markov graph is equivalent to the dynamic $ER(p)$ graph.
\end{observation}

\begin{observation}
Another special case is the $(p,p)$-stochastic model. Here, define $p$ to be the {\em stability factor}. For small $p$, there are few changes from $G_t$ to $G_{t+1}$ and the graph is {\em stable}. For large $p$, there could be many changes from $G_t$ to $G_{t+1}$ and the graph is {\em unstable}. A special case of this special case is the $(1,1)$-stochastic model in which edges and non-edges alternate at each time slot.
\end{observation}

\section{Analyzing Latency along Dynamic Paths}
\label{sec:latency}

Many routing schemes determine a path (say, according to a shortest path calculation), and then stay on that path even though it may be intermittently connected due to edges on it appearing and disappearing according to one of the aforementioned stochastic processes.

Hence we consider the simplest case which is amenable to mathematical analysis -- the underlying graph $G_u = L_n$, the line graph with $n$ vertices and $n-1$ edges in which vertex $1$ wants to send a message to vertex $n$. We denote these graphs by $ER(n,p,L_n)$ and $MC(q,p,L_n)$. Clearly a message should either be stored or be either advanced as much as possible (under the {\sf CuT} model) or one hop per time slot (under the {\sf SoA} model).

We now study how random variables such as time taken to reach node $n$ from node $1$ behave as a function of $n, p, q, G_u$. We first show how simple expected value analysis can yield first moments, and then characterize the entire probability distributions as a function of such parameters. The results of this analysis will be applicable to the analysis of Temporal Graphlets in Section \ref{sec:stgsmg}.

\subsection{The $(1,1)$-Stochastic Model}

For the $(1,1)$-stochastic model, one can compute the exact arrival time. Define a {\em configuration} as a binary string of length $n-1$. If the $i$-th bit is $1$ then the $i$-th edge on the line exists otherwise it does not exist. For a given binary string $B$, let $k(B)$ be the number of changes from $0$ to $1$ or from $1$ to $0$ and let $b(B)$ be the value of the first bit of $B$. For example, $k(001110011001)=5$.

\begin{observation}
The routing in the {\sf CuT} model takes $k+1-b$ slots.
\end{observation}

\begin{observation}
The routing in the {\sf SoA} model takes $2n-k-b$ slots.
\end{observation}

\begin{corollary}
The best configuration for {\sf CuT} is $111\cdots 1$ for which the routing takes $0$ slot\footnote{This assumes that cutting through the network takes negligible time compared to waiting.}
\end{corollary}

\begin{corollary}
The worst configuration for {\sf CuT} is $0101\cdots$ for which the routing takes $n-1$ slots.
\end{corollary}

\begin{corollary}
The best configuration for {\sf SoA} is $1010\cdots$ for which the routing takes $n-1$ slots.
\end{corollary}

\begin{corollary}
The worst configuration for {\sf SoA} is $000\cdots 0$ for which the routing takes $2(n-1)$ slots.
\end{corollary}

We now compute the average routing time assuming a uniform distribution for all the $2^{n-1}$ configurations.

\begin{observation}\label{obs:pq1CuT}
The average routing time for {\sf CuT} is $\frac{1}{2}(n-1)$ slots.
\end{observation}

\begin{observation}\label{obs:pq1SoA}
The average routing time for {\sf SoA} is $\frac{3}{2}(n-1)$ slots.
\end{observation}

\subsection{The $(1-p,p)$-Stochastic Model}
\label{sec:ERline}

This is equivalent to the $ER(n,p,L_n)$ model. We first begin with computation of expected values of advancement of a message until it hits a non-edge and the expected routing latency. Subsequently we derive the exact probability distributions of the spatio-temporal location of the message as well the distribution of the routing latency under both the {\sf SoA} and {\sf CuT} models.

\begin{observation}
In {\sf SoA} the expected advance is $p\cdot 1+(1-p)\cdot 0=p$. 
\end{observation}

\begin{observation}
In {\sf CuT} the expected advance is upper-bounded by $(1-p)\sum_{i=1}^{\infty}ip^i=\frac{p}{1-p}$.
\end{observation}

The following corollaries follows since the length of the route is $n-1$.

\begin{corollary}\label{cor:SoA}
In {\sf SoA} the expected time for the routing time is $\frac{n-1}{p}$.
\end{corollary}

\begin{corollary}\label{cor:CuT}
In {\sf CuT} the expected time for the routing time is $\frac{(n-1)(1-p)}{p}$.
\end{corollary}

\paragraph{{\sf SoA} latency}
Consider an Erdos-Renyi line graph on $n$ nodes which denoted by $ER^t(n,p,L_n)$ at the $t$-th time instant. There are a maximum of $n-1$ edges in this graph, and at each time instant, each edge exists with probability $p$. We want to send a packet from node $1$ to node $n$; if an edge $(u,v)$ is up at time instant $t$, and $u$ has the packet, then it will transmit to $v$ in that instant, otherwise, it will hold it until a later time instant when the edge becomes active. We want to track the probability distribution of the packet over time as a function of $n$ and $p$.

Let $N_t$ be a random variable denoting the node that the packet has reached at time $t$, and $e_k$ be the $k$-th edge.
\squeeze
\setlength{\arraycolsep}{0.0em}
\begin{eqnarray}
\nonumber P(N_t=k) & = & P(N_{t-1}=k-1) P(e_{k-1}) + P(N_{t-1}=k) P(\overline{e_k}) \\
\label{eq:probdistER} & = & P(N_{t-1}=k-1) p + P(N_{t-1}=k) (1-p)
\end{eqnarray}
\setlength{\arraycolsep}{5pt}

It is difficult to solve the above bivariate recurrence to attain a closed form for $P(N_t=k)$, hence we compute the probabilities numerically. Figure \ref{fig:Ntk} shows an example of a probability distribution for a small line graph. The example considers a line graph on $n=10$ nodes for $p=0.25$. It is easy to see that since the expected waiting time for every hop is $\frac{1}{p}$, each hop takes approximately 4 time slots to traverse. Hence at $t=20$, the packet would have traversed a mean of 5 hops, which is indicated in the figure.

\begin{figure}[!t]
\centering
\includegraphics[width=0.5\textwidth]{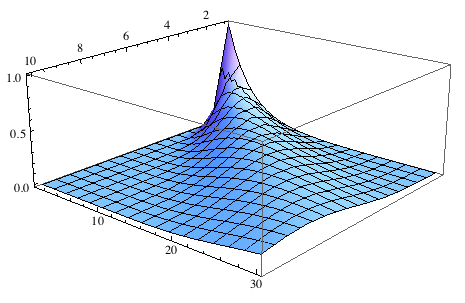}
\caption{\label{fig:Ntk} Probability distribution of the packet as function of space and time for $ER(n=10,p=0.25,L_n)$ for $t=30$}
\end{figure}

Let $T$ be a random variable denoting the number of time slots needed for a packet to reach from node $1$ to node $n$. It is easy to see that $P(T<n-1)=0$ since it takes at least $n-1$ slots to reach node $n$. The general distribution of $T$ is given by the following:
\squeeze
\begin{equation}\label{eq:timedistER}
P(T=n-1+j)={n+j-2 \choose j} (1-p)^j p^{n-1}, \forall{j} \geq 0
\end{equation}

This is because there are exactly $j$ time slots when the packet has to wait at one of the nodes $1,2,3,\ldots,n-1$, and there are ${n+j-2 \choose j}$ number of ways of assigning these $j$ slots to the $n-1$ nodes. Figure \ref{fig:timeLnERp} plots this distribution.

It can easily be verified that $E[T]=\sum_{k=0}^\infty k P(T=k) = \frac{n-1}{p}$, which is in agreement with Corollary \ref{cor:SoA}. 

\paragraph{{\sf CuT} latency}
We now characterize the distribution of routing times in terms of the {\em cut-through} metric. It is assumed that the time taken to {\em cut through} the edges in a connected component do not cost any time slots and time elapses only due to waiting for an inactive link to become active\footnote{A useful metaphor would be that of light passing through an intermittently connected network. The time scales of disruption are much lower than those of light traversing a connected component.}. 

Let $T$ be the random variable denoting the number of time slots taken to reach node $n$ from node $1$ if nodes were forwarding the packet as much as possible toward the destination in the current connected component. 
\squeeze
\setlength{\arraycolsep}{0.0em}
\begin{eqnarray}
\nonumber P(T=k) &=& Pr\{ \textrm{Wait for $k$ slots at $\{1,2,\ldots,n-1\}$} \}\\
\label{eq:erprob} &=& {n+k-2\choose k} (1-p)^k p^{n-1}
\end{eqnarray}
\setlength{\arraycolsep}{5pt}

\begin{figure}[!t]
\centering
\includegraphics[width=0.45\textwidth]{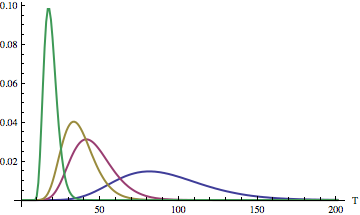}
\caption{\label{fig:timeLnERp} Probability distribution of time taken to traverse the dynamic line graph $ER(10,p,L_{10})$ for values of $p=\{0.1,0.2,0.25,0.5\}$}
\end{figure}

This is because the number of ways of assigning $k$ waiting slots at one or more of nodes $\{1,2,\ldots,n-1\}$ is the same as number of ways putting $k$ balls in $n-1$ distinct bins with no restrictions on the number of balls in a particular bin, and this is given by ${n+k-2\choose k}$. Note that the only reason the packet needs to wait for a slot at node $j$ is if the edge $(j,j+1)$ is inactive at that time instant. This contributes to the $p^k$ term.

It can be verified that $E[T] = \sum_{k=0}^\infty k P(T=k) = (n-1)\frac{1-p}{p}$, which is consistent with Corollary \ref{cor:CuT}. Also, the variance is given by: $Var[T] =  E[T^2]-E[T]^2=(n-1)\frac{1-p}{p^2}$. Not surprisingly the mean time elapsed when using the {\sf CuT} metric is smaller than that in case of the {\sf SoA} metric.

\subsection{The $(q,p)$-Markov Model}

Now we study routing on dynamic line graphs $MC(p_0,q,p,L_n)$, where $p_0$ is the probability of an edge existing in the first graphlet. 

\begin{observation}
It is easy to see that this Markov chain has a stationary distribution $\pi=(\pi_{on},\pi_{off})=(\frac{p}{p+q},\frac{q}{p+q})$.  To eliminate the effect of transients, we assume that the Markov chain has converged (or {\em mixed}) before node $1$ sends the message to node $n$; in other words, $p_0=\pi_{on}=\frac{p}{p+q}$.
\end{observation}

\begin{observation}
In {\sf CuT} the expected advance on an infinite line is upper-bounded by $\frac{p}{q}$.
\end{observation}

\begin{observation}
In {\sf SoA} the expected advance is $\frac{p}{p+q}$.
\end{observation}

\begin{corollary}\label{cor:CuTpq}
In {\sf CuT} the expected time for the routing time is $\frac{(n-1)q}{p(p+q)}$. [Proof omitted]
\end{corollary}

\begin{corollary}\label{cor:SoApq}
In {\sf SoA} the expected time for the routing time is $n-1+\frac{(n-1)q}{p(p+q)}$. [Proof omitted]
\end{corollary}
 
\eat{
Note that for the case $p=q=1$ we get here different results than those we had before. In {\sf SoA} we get $2n$ instead of $3n/2$ and in {\sf CuT} we get $n$ instead of $n/2+1$. The reason is that there we assumed a uniform distribution over all configurations and here we assumed a steady state configuration.
}

\begin{figure}[!t]
\centering
\includegraphics[width=0.45\textwidth]{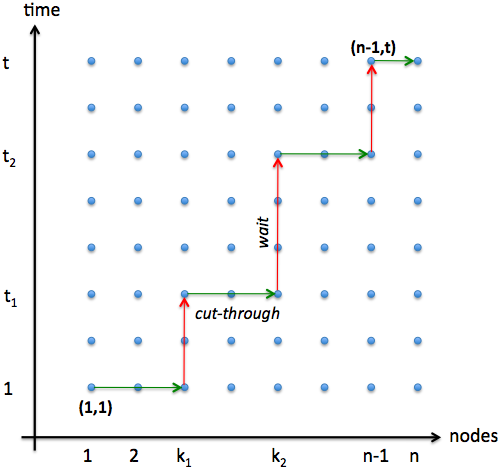}
\caption{\label{fig:timespace} Analyzing the {\sf CuT} latency in a dynamic random line graph}
\end{figure}

\paragraph{{\sf CuT} latency}
Figure \ref{fig:timespace} illustrates a sample path from $1\rightarrow n$ over time\footnote{We present {\sf CuT} before {\sf SoA} since the former is easier to explain, and we will reuse the analysis technique for the latter, later on.}. There are several such paths possible depending on the state of the edges, and the computation here is more involved than the $ER(n,p,L_n)$ case.

Any path through this space-time can be characterized by its constituent segments:
$\{1, k_1, t_1, k_2, t_2, \ldots, k_m, t_m, n\}$, where $t_m=t$. Clearly $1\leq m \leq \min(n-1,t-2)$.

Let $X_\tau^e$ correspond to a binary random variable that denotes the status of edge $e$ at time instant $\tau$. The probability that path $P = \{1, k_1, t_1, k_2, t_2, \ldots, k_m, t_m, n\}$ exists is given by the following:
\squeeze
\setlength{\arraycolsep}{0.0em}
\begin{eqnarray}
\label{eq:mc1} Pr\{P\} & = & Pr\{X_1^1,X_1^2,\ldots,X_1^{k_1-1},\overline{X}_1^{k_1},\overline{X}_2^{k_1},\ldots,\overline{X}_{t_1-1}^{k_1},X_{t_1}^{k_1},\nonumber\\
& & \ldots, X_t^{k_m},\ldots,X_t^{n-1} \}\\
\label{eq:mc2} & = & P(X_1^1)\cdots P(X_1^{k_1-1}) P(\overline{X}_1^{k_1},\overline{X}_2^{k_1},\ldots,\overline{X}_{t_1-1}^{k_1},X_{t_1}^{k_1}) \nonumber\\
& & \cdots P(X_t^{k_m}) P(X_t^{n-1}) \\
\label{eq:mc3} & = & \pi_{on}^{k_1-1} \pi_{off} (1-p)^{t_1-2} p \times \nonumber \\
& & \pi_{on}^{k_2-k_1-1} \pi_{off} (1-p)^{t_2-t_1-1} p \times \nonumber\\
& & \cdots \pi_{on}^{n-k_m-1}\\
\label{eq:mc4} & = & \pi_{on}^{n-m-1} \pi_{off}^m (1-p)^{t-m-1} p^m\\
& = & (\frac{p}{p+q})^{n-m-1} (\frac{q}{p+q})^m (1-p)^{t-m-1} p^m\\
& = & \frac{p^{n-1} q^m (1-p)^{t-m-1}}{(p+q)^{n-1}}, \quad \textrm{ where } p>0, q>0
\end{eqnarray}
\setlength{\arraycolsep}{5pt}

Equation \ref{eq:mc2} follows from Eq. \ref{eq:mc1} by using the fact that probabilities of statuses of various edges are independent of each other. However, the probability of existence of an edge (say $k_1$) at successive time instants are related by the Markov chain parameters, $p$ and $q$. Therefore, we have:
\squeeze
\setlength{\arraycolsep}{0.0em}
\begin{eqnarray}
\lefteqn{P(\overline{X}_1^{k_1},\overline{X}_2^{k_1},\ldots,\overline{X}_{t_1-1}^{k_1},X_{t_1}^{k_1})} \nonumber\\
&=& P(\overline{X}_1^{k_1}) P(\overline{X}_2^{k_1}|\overline{X}_1^{k_1}) \cdots P(X_{t_1}^{k_1}|\overline{X}_{t_1-1}^{k_1}) \nonumber \\
\label{eq:mc5} &=& \pi_{off} \: (1-p)^{t_1-2} \: p
\end{eqnarray}
\setlength{\arraycolsep}{5pt}

For each segment corresponding to {\em waiting}, the probability of the existence of that segment is given by Equation \ref{eq:mc5}. Using the fact that there exist $m$ such "wait" segments and $m$ "cut-through" segments, Eq. \ref{eq:mc4} can be simplified from Eq. \ref{eq:mc3}\footnote{We note that this technique can be used in the probability computation for the case where each edge $e$ has a different $(q_e,p_e)$. The expression \ref{eq:mc4} will then exhibit a much more complicated product form.}.

Let the number of paths that have exactly $m$ bends be $N_m$. We observe that a path may be generated by independently choosing $m$ bending points each on the space and time axes. The number of ways of doing so are ${n-1\choose m}$ and ${t-2\choose m-1}$ respectively. Hence $N_m ={n-1\choose m}{t-2\choose m-1}$. Therefore, the latency probability distribution for $p>0,q>0$ is given by:
\squeeze
\setlength{\arraycolsep}{0.0em}
\begin{eqnarray}\label{eq:mcprob}
\lefteqn{P(T=t-1)} \nonumber\\
& = & \sum_{m=1}^{n-1} {n-1\choose m}{t-2\choose m-1} \frac{p^{n-1} q^m (1-p)^{t-m-1}}{(p+q)^{n-1}}
\eat {\\ & = & (n-1) \frac{p^{n-1} q (1-p)^{t-2}}{(p+q)^{n-1}} {}_2F_1 (n-2,t-2;2;\frac{q}{1-p}) }
\end{eqnarray}
\setlength{\arraycolsep}{5pt}

\eat{
where ${}_2F_1 (a,b;c;z)$ is the hypergeometric function:
\[ {}_2F_1 (a,b;c;z) = 1 + \frac{ab}{1! c}z + \frac{a(a+1)b(b+1)}{2! c(c+1)}z^2 + \cdots \]
}

If $p+q=1$, the Markov chain reduces to the {\em independent} $ER(n,p,L_n)$ scenario, and it can be verified that Equation \ref{eq:mcprob} reduces to Equation \ref{eq:erprob} (with $k$ substituted for $t-1$). We have also numerically verified that for $p\rightarrow 1, q\rightarrow 1$, $E[T]=\frac{1}{2}(n-1)$, in agreement with Observation \ref{obs:pq1CuT}, and that the general case is in agreement with Corollary \ref{cor:CuTpq}.

\paragraph{{\sf SoA} latency}

\begin{figure}[!t]
\centering
\includegraphics[width=0.45\textwidth]{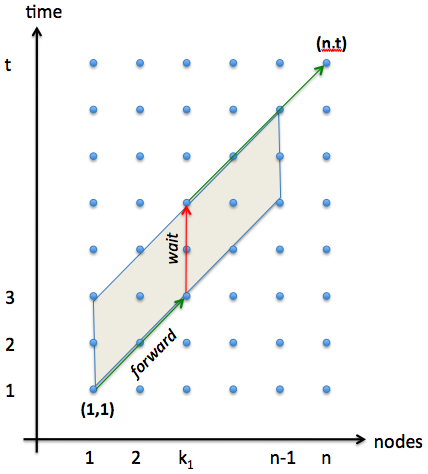}
\caption{\label{fig:timespace-soa} Analyzing the {\sf SoA} latency in a dynamic random line graph}
\end{figure}

Figure \ref{fig:timespace-soa} illustrates the latency under the {\sf SoA} model. Since each forwarding action to the neighbor costs a time slot, the latency $t-1$ obeys $t-1 \geq n-1$, with the best case scenario being the diagonal green path from $(1,1)$ to $(n,t)$. Hence if we want to compute $P(T=t-1)$, we have to consider all paths that use the diagonal "forward" segments and vertical "wait" segments, and are contained in the shaded parallelogram; these segments  eventually reach $(n,t)$. The width of this parallelogram is $t-n$.

We borrow the techniques used in the {\sf CuT} probability computation previously and note that paths with $m$ waiting points are possible inside this parallelogram with $1\leq m \leq \min(n-1,t-n)$. Using similar techniques as the {\sf CuT} computation, we can compute the probability of a certain path $P$ inside the parallelogram with $m$ waiting points (or ``bends") as follows:
\setlength{\arraycolsep}{0.0em}
\begin{eqnarray}
\label{eq:mcsoa} Pr\{P\} & = & \pi_{on}^{n-m-1} \pi_{off}^m (1-p)^{t-n-m} p^m
\end{eqnarray}
\setlength{\arraycolsep}{5pt}

Let the number of paths that have exactly $m$ waiting points (or ``bends") be $N_m$. Since a path may be generated by independently choosing $m$ bending points each on the diagonal and vertical axes of the parallelogram, $N_m ={n-1\choose m}{t-n-1\choose m-1}$. Therefore, the latency probability distribution for $p>0,q>0$ is given by:
\squeeze
\setlength{\arraycolsep}{0.0em}
\begin{eqnarray}\label{eq:mcsoaprob}
\lefteqn{P(T=t-1)} \nonumber\\
\nonumber & = & \sum_{m=1}^{\min({n-1,t-n})} {n-1\choose m}{t-n-1\choose m-1} \frac{p^{n-1} q^m (1-p)^{t-n-m}}{(p+q)^{n-1}}\\
& &
\end{eqnarray}
\setlength{\arraycolsep}{5pt}

where $t \geq n$. If $p+q=1$, the Markov chain reduces to the {\em independent} $ER(n,p,L_n)$ scenario, and it can be verified that Equation \ref{eq:mcsoaprob} reduces to Equation \ref{eq:timedistER}. We have also numerically verified that for $p\rightarrow 1, q\rightarrow 1$, $E[T]=\frac{3}{2}(n-1)$, in agreement with Observation \ref{obs:pq1SoA}, and that the general case is in agreement with Corollary \ref{cor:SoApq}.

\section{Stacked and Smashed Representations of Temporal Graphlets}
\label{sec:stgsmg}

Since there is  a solid theory of traditional non-temporal graphs, an obvious question to ask is if the study of some temporal properties may be reduced to studying the same property on an equivalent single non-temporal graph. We consider two such representations -- the {\em stacked graph (StG)} and the {\em smashed graph (SmG)}. A stacked graph is constructed by drawing directed edges in the direction of time between successive temporal graphlets in a TGS; a smashed graph is a ``collapsed'' version of the stacked graph. Alternatively, it is union of the TGs. Clearly, an SmG is a ``lossy'' version. However, it is far more succinct, and therefore it would be interesting to know
when, if at all, it will suffice.

The study of such ``reducibility'' is  helpful in that it will allow us to use well-known graph-theoretic algorithms (and code) on the appropriate representation to easily evaluate whether properties such as reachability, connectivity etc. hold.

We note that the ``evolving graph'' representation proposed in  \cite{Ferreira04} which labels edges with the times at which they are active is equivalent to the stacked graph\footnote{And deleting the labels yields a smashed graph.}  but an evolving graph not a traditional graph. Hence reducing to an evolving graph does not allow us to easily leverage existing algorithms or code. It is imaginable that a smashed graph (or its $m$-smashed variant defined later) can be used to quickly answer {\em on-line} queries for graph properties in massive temporal graphs even though such queries may only be answered approximately. Therefore, it is interesting and worthwhile to compare the complexity vs. accuracy tradeoffs of smashing for various temporal graphs.

\begin{figure}[!t]
\centering
\includegraphics[width=0.45\textwidth]{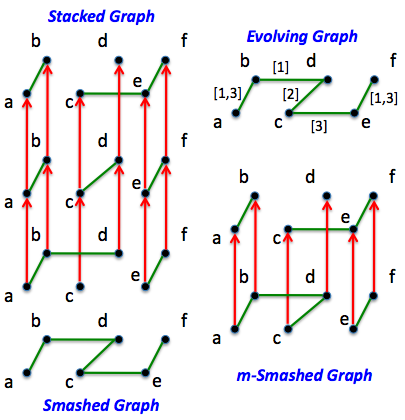}
\caption{\label{fig:tgreps} Various representations of temporal graphlets for the TGS in Fig. \ref{fig:graphlets}}
\end{figure}

\subsection{Definitions and Basic Properties}

We begin with some definitions.

\begin{definition}
Given a temporal graphlet sequence $G[1,T]$, the {\em stacked graph (StG)} of $G[1,T]$ is $StG = (V_S, E_S)$, where $V_S = \cup_t V(t)$, $E_S = \cup_t E(t) \cup E_C$ where $E_c$ is a set of ``cross edges'' connecting vertices of adjacent (in time) graphlets. That is, $E_c = \cup_{t,i} (u_i(t), u_i(t+1))$.
\end{definition}

\begin{definition}
Given a temporal graphlet sequence $G[1,T]$, the {\em smashed graph (SmG)} of $G[1,T]$ is $SmG = (V_M, E_M)$, where each sequence of $u(t), u(t+1), \ldots $ is replaced by a single vertex $u \in V_M$, and $E_M = \cup_t E(t)$ with endpoints of edges mapped to the replaced vertices in $V_M$.
\end{definition}

\begin{definition}
Given a temporal graphlet sequence $G[1,T]$, the {\em $m$-smashed graph (m-SmG)} of $G[1,T]$ is $m$-$SmG = (V_M, E_M)$, where the smashing operation is not performed on the entire $G[1,T]$ but on each of $G[1,m], G[m+1,2m], G[2m+1,3m], \ldots$ instead.
\end{definition}

The various aforementioned representations of the temporal graphlet sequence shown in Figure \ref{fig:graphlets} are illustrated in Figure \ref{fig:tgreps}. As mentioned earlier, the StG and Ferreira's evolving graph model are equivalent in terms of information content. On the contrary SmG is lossy since temporal ordering information is lost during smashing of graphlets. This can result in some false positives (e.g., in the smashed graph, $e\rightarrow b$ is a valid spatio-temporal path, whereas that is not the case in reality). 

The technique of $m$-smashing tries to balance the tradeoffs between StG (or evolving graphs) and SmG by restricting the smashing to a smaller number of graphlets at a time. For example, in Figure \ref{fig:tgreps}, the first two graphlets are smashed into one, and the result is stacked with the third graphlet. Note that some false positives that were deduced from the SmG (e.g., $e\rightarrow b$ and $e\rightarrow d$) disappear in $m$-SmG. However, some other false positives such as $c\rightarrow b$ still remain.

We note that StG and SmG are non-temporal, or traditional graphs. Consider a property $P$ (definitions of some basic properties studied in this paper are in Table \ref{tab:properties}). Can the question of whether $P$ is true in $G[1,T]$ be answered by evaluating $P$ on $StG$? If we can, we call such a  property {\em stacked-graph reducible (StG-reducible)}. Similarly, if it can be answered by evaluating $P$ on $SmG$ then we call it {\em smashed-graph reducible (SmG-reducible)}.

\begin{definition}
Let $P(H)$ be a function denoting the value (including true/false) of a property P on a structure H where H could be a temporal graphlet sequence or a graph. Then, property $P$ is {\em StG-reducible} iff $P(G[1,T]))$ = $P(StG))$, and $P$ is {\em SmG-reducible} iff $P(G[1,T]))$ = $P(SmG))$.
\end{definition}

\begin{table*}[tbp]
\centering
\caption{\label{tab:properties} Examples of temporal graph properties}
\begin{tabular}{|l|l|}\hline
\textbf{T-* Property} & \textbf{Definition}\\ \hline
T-adjacent$(u,v)$ & $\exists (u,v) \in G[1,T]$ \\ \hline
T-reachable$(u,v)$ & $\exists \{(u,v_1),\ldots,(v_k,v)\} \in G[1,T]$\\ \hline
T-clique & $\max\{X | X \subset V(1), \forall_{u,v} \in X, \textrm{ T-adjacent}(u,v)\}$\\\hline
T-$k$-connected & $\forall S = \{v_1,\ldots,v_{k-1}\}\in V_S$, if $S$ is removed, \\
& $\forall_{u,v}$  T-reachable$(u,v)$\\ \hline
\end{tabular}
\end{table*}


We first consider StG-reducibility. We note that some properties such as clique are not ``well formed'' for directed graphs. In such a case, we admit the use of the undirected version, that is, if $P$ is evaluated on $StG$ by simply ignoring the direction of the edges. We now consider a few properties.

\begin{observation}\label{obs:treach}
T-reachability is StG-reducible. This is because the cross edges are tantamount to the ``store'' action. 
\end{observation}

\begin{observation}\label{obs:tclique}
T-clique is not StG reducible.
\eat{This is because for it to be so, we need cross edges from every node in G(t) to every node in G(t+1)}.
\end{observation}

\begin{observation}
T-$k$-connectivity is StG-reducible if and only if $k$ = 1.
\end{observation}
\begin{proof}
That 1-connectivity is StG-reducible 
follows by repeated application of Observation \ref{obs:treach}.
That 2-connectivity is not SG-reducible is illustrated by the ``temporal
triangle'' which is defined as follows:
$V(1) = V(2) = V(3) = \{a, b, c\}$, $E(1) = (a,b)$, $E(2) = (b,c)$,
$E(3) = (c,a)$.
G[T] is 2-connected, but in $G_S$ the cross edge between
c(1) and c(2) is a bridge. It is easy to see that this is extensible to
3-connectivity in temporal $K_4$ and so on.
\end{proof}

\eat{
\begin{observation}
T-chromatic number is StG-reducible
\end{observation}
\begin{proof}
Let $\chi(G)$ denote the chromatic number of a graph or TGS. 
We need to show that $\chi(StG) = \chi(TGS)$. This is clearly true if
$E(t) = \phi$ for all $t$, so suppose there is a 
$t = T$ after which $|E(t)|$ $\geq$ 1. Therefore $\chi(TGS) \geq 2$.
Now consider the TGS after this $T$ and its StG equivalent. Consider a
modified version of the StG into StG$^T$ by introducing a ``dummy''
node on each cross edge.
Clearly, since we can assign any already-used 
color not assigned to the two successive
nodes in the sequence to the dummy node, and there must be such a color
(since chromatic number is at least 2), $\chi(StG^T) = \chi(StG)$. 
 Note that in StG$^T$ the dummy nodes serve to 
preserve the color of nodes as we go forward in time, thus the addition of
an edge at some $t$ will equally increase the chromatic
number of both the StG$^T$ and the TGS. Deletion of an edge does
not change the chromatic number per definition of T-chromatic number.
Thus, $\chi(TGS) = \chi(StG^T)$. It follows that $\chi(TGS) = \chi(StG)$.
\end{proof}
}

We now consider Smashed Graphs (SmG) and SmG-reducibility. Since SmG is
lossy, it is clear that for the arbitrary case it is not reducible. However,
there are two questions: 1) how close can we come? 2) are there special cases
when it is reducible? The first question is the subject of later sections,
here we state some simple results.




\begin{figure}[!t]
\centering
\includegraphics[width=0.5\textwidth]{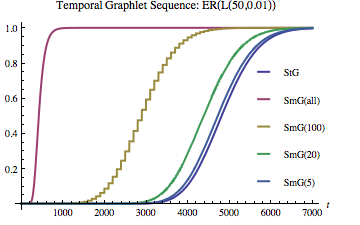}
\caption{\label{fig:stgsmganal} Probability distribution of reachability on a $ER(n,p,L_n)$ TGS. SmG(all) corresponds to SmG and SmG($m$) correspond to $m$-SmG.}
\end{figure}

\begin{observation}
\eat{T-chromatic number and T-clique are SmG-reducible.}
T-clique is SmG-reducible.
\end{observation}

\begin{observation}
\label{SmGReducibility}
T-reachability is SmG-reducible if either of the following holds: (a)
there is some G(t=T) that is identical to SmG; (b) there do not exist
G(t) and G(t+1) such that the number of connected components increases.
\end{observation}

Consider observation~\ref{SmGReducibility}(a) when the identical graphlet either
occurs as the first or last in the sequence.

\begin{corollary}
T-reachability is SmG-reducible if either (a) no edges are ever added;
(b) no edges are ever deleted.
\end{corollary}

Observation \ref{SmGReducibility} allows arbitrary additions and
deletions, but in a manner that preserves reachability. In practice,
these conditions are easily checkable on a sequence of
TGs and if they ``pass'', we can use the SmG as a way to get the value of
graph theoretic properties such as reachability, clique, etc.

The StG- and SmG-reducibility of numerous other graph theoretic concepts is
interesting and open.

\subsection{Probabilistic Analysis of Smashing}

We analyze the properties of stacking and smashing on a random TGS constructed from a sequence of random Erdos-Renyi line graphs given by $\{ER(n,p,L_n)\}_{t=1}^T$.

The probability that a path of latency $t$ time slots (under {\sf CuT} metric) exists from node $1$ to $n$ is given by Eq. \ref{eq:erprob}.
\eat{
\[ P(T=t) = {n+t-2\choose t} (1-p)^t p^{n-1} \]
}
Hence the probability that node $1$ can reach node $n$ {\em within} $T=t$ graphlets is given by:
\squeeze
\begin{equation}
P(T< t) = \sum_{\tau=0}^{t-1} {n+\tau-2\choose \tau} (1-p)^\tau p^{n-1}
\end{equation}

\eat{
\begin{equation} 
= 1-p^{n-1}(1-p)^t \frac{\Gamma_{n+t-1}}{\Gamma_{n-1}\Gamma_{t+1}} {}_2F_1 (1,n+t-1;t+1;1-p)
\end{equation}

where ${}_2F_1 (a,b;c;z)$ is the hypergeometric function mention earlier.
}
\eat{\[ {}_2F_1 (a,b;c;z) = 1 + \frac{ab}{1! c}z + \frac{a(a+1)b(b+1)}{2! c(c+1)}z^2 + \cdots \]}

If all the $T$ graphlets are smashed into a single graph, then we can compute the probability of existence of a path from node $1$ to $n$ on the {\em smashed} graph SmG. $(j,j+1)\in SmG$ iff $\exists G_i \in \{G_1, G_2, \ldots, G_T\}\:\textrm{s.t.}\:(j,j+1)\in G_i$. The probability of this happening is given by:
\squeeze
\begin{eqnarray}
\nonumber Pr\{(j,j+1)\in SmG\} & = & 1-Pr\{(j,j+1)\notin SmG\}\\
\nonumber  & = & 1-(1-p)^t
\end{eqnarray}

Therefore, the probability that a path exists within $T$ graphlets is given by the following (since all $n-1$ edge probabilities are independently distributed):
\squeeze
\begin{equation}\label{eq:tsmg} 
P(T_{SmG} < t) = (1-(1-p)^t)^{n-1} 
\end{equation}

If we decide to smash $m$ graphlets at a time into one but preserve the rest of the stacked structure, then we have $\frac{T}{m}$ graphlets instead of $T$. The probability of existence of an edge in any of these smashed graphlets is $1-(1-p)^m$. Hence the probability distribution of the existence of a path in an $m$-smashed TGS is given by:
\squeeze
\begin{equation}\label{eq:tsmgm}
P(T_{SmG}^{(m)}< t) = \sum_{\tau=0}^{\frac{t}{m}-1} {n+\tau-2\choose \tau} (1-p)^{m\tau} (1-(1-p)^m)^{n-1} 
\end{equation}

Figure \ref{fig:stgsmganal} illustrates the probability distributions of stacked, smashed, and $m$-smashed graphlets. We can observe that while smashed graphs yield only a crude upper bound on the real probabilities (i.e. stacked graphs), the procedure of $m$-smashing is useful since it can yield probability distributions that are much better upper bounds especially for low values of $m$. For graphs where there exist multiple potential paths between source and destination, this process is likely to be even more useful.

\begin{figure}[!t]
\centering
\includegraphics[width=0.5\textwidth]{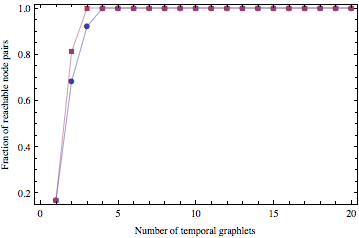}
\caption{\label{fig:stgsmg} Fraction of reachable node pairs in $MC(n=20,p_0=0.005,p=0.5,q=0.05)$ TGS. Squares correspond to SmG and circles to StG}
\end{figure}

While we have only shown the $ER(n,p,L_n)$ scenario for the {\sf CuT} scenario here, it is easy to extend it to the $MC(n,p_0,p,q,L_n)$ scenario, since one can apply Eq. \ref{eq:mcsoaprob} in this setting to compute the probability of existence of a path within $T=t$ units of time. The probability of existence of an edge in a smashed graphlet in this model can is given by:

\begin{equation}
P(T_{SmG} < t) = (1-\frac{q}{p+q}(1-p)^{t-1})^{n-1} 
\end{equation}

We omit further details due to paucity of space.

The effect of smashing was also investigated on another derivative metric, namely, the number (or fraction) of reachable pairs over a given time budget $T$ for $G_u=K_n$. We found by simulations that while the gap between SmG an StG was large for the $ER(n,p)$ scenario (hence motivating $m$-smashing as shown earlier), for the $MC(p_0,q,p)$ scenario, there were parameter values for which the gap was much lower (as exemplified by Fig. \ref{fig:stgsmg}). A thorough analysis of the $(q,p)$ parameter space with respect to the reachability metric is a topic of future research.

\section{Adaptive Routing in Dynamic Random Graphs}
\label{sec:routing}

Traditional shortest paths problems attempt to find a path of minimum total distance from $s$ to $t$ and are solved by classical algorithms that satisfy the suboptimal path property~\cite{Bertsekas92} (e.g., Dijkstra's and Bellman-Ford). Shortest paths, routing, and related problems have been considered in various stochastic models (when edges disappear {\em permanently}, when they do so {\em periodically}, etc.; see \cite{ChawlaR06,Bar-NoyS91,Jain04,BalasubramanianLV07,dutycycledSP,KonjevodOR02,DemetrescuI06}) but to the best of our knowledge, have not before been studied in the model considered here.
Specifically, we consider the $ER^t(n,p,G_u)$ model with {\sf SoA}. In the adaptive generalization of the shortest paths concept to temporal graphs (sometimes called ``next-hop routing''), the task is to choose, at each routing stage, the best neighbor to route to, if any, in order to minimize the remaining expected travel time\footnote{In some cases, it may be best to remain at the current node for another birth/death time step, and then reevaluate.}. We solve the problem optimally, using a variant of Dijkstra's algorithm. Because in the adaptive setting we make a routing decision adaptively, each time we arrive at a node, based on its current set of outgoing edges, the algorithm performs its computation going from $t$ to $s$, rather than from $s$ to $t$.
To motivate the algorithm, we make the following observations.

\begin{observation}
In an unweighted graph, an optimal move from a neighbor $v$ of $t$ is to remain at $v$ until the edge $(v,t)$ appears, and then traverse it.
\end{observation}
\begin{proof}
Since traversing the edge takes one birth/death time step, the likelihood of being able to traverse edge $(v,t)$ at the next time step is the same as that of being able to traverse $(v',t)$ for some mutual neighbor $v'$.
\end{proof}

\begin{corollary}\label{nn}
In a weighted graph, an optimal move from $t$'s nearest neighbor $\hat{v}$ will be to remain at $\hat{v}$ until the edge $(\hat{v},t)$ appears.
\end{corollary}

\begin{observation}
In an optimal adaptive routing path, there will without loss of generality be no backtracking.
\end{observation}
\begin{proof}
Suppose in an optimal solution, we move from node $v$ to $u$. Assume we move only when it gives us a strict improvement, i.e., $METT[u] < METT[v]$. In that case, once at $u$, we will never move to a node with expected remaining travel time greater than $u$'s.
\end{proof}

Given this, the optimal deterministic routing algorithm simply moves greedily in order to decrease the remaining minimum expected traversal time ($METT$): at time step, move from the current node $v$ to a neighbor $u \in N(v)$ of minimum $METT(u)$ from $t$, if there is one such that $METT(u) < METT(v)$; otherwise, remain at node $v$ until the next time step.
This algorithm assumes an oracle to compute $METT(v)$ for each node $v$, which is done by Algorithm \ref{times}.


\begin{algorithm}[h!]
\caption{Computing minimum expected traversal times} \label{times}
\begin{algorithmic} [1]

\FOR {each $v$}
   \STATE $METT[v] = \infty$
\ENDFOR

\STATE $METT[t] = 0$
\STATE $Q \leftarrow V$

\WHILE {$Q \ne \emptyset$}

   \STATE $u \leftarrow$ extract-min$(Q)$
   \FOR {each $v \in N(u) \cap Q$} 
      \STATE $d_v \leftarrow f(p,\{METT[\hat{u}] : \hat{u} \in N(v)-Q\})$
      \IF {$d_v < METT[v]$}
         \STATE $METT[v] = d_v$
      \ENDIF
   \ENDFOR
\ENDWHILE

\end{algorithmic}
\end{algorithm}

Let $N(v)$ indicate the neighbors of $v$, and $f(p,S)$ indicate the function computing the $METT$ from a node $v$ to $t$, along a {\em path whose next-hop node is a member of $S$}. Given the expected traversal times of the nodes in $S$, $f$ can easily be computed in time $O(|N(v)| \log |N(v)|)$: sort the neighbor nodes $v$ in order of $METT[v]$.
The set of neighbor nodes chosen to consider as next-hop candidates in the event that we arrive at node $u$ is the prefix of the $v$ sequence that minimizes the expected remaining traversal time from $v$ to $t$.

\begin{lemma}
Restricted to an available set of nodes $S$ to use as next hop, and based on the correct $METT$ values of the members of $S$, the function $f(p,S)$ correctly computes $METT[u]$.
\end{lemma}
\begin{proof}
We need a policy that tells us, when offered a
set of the choices, which we should accept, if any, or whether we
should instead wait a timestep and try again. Since the graph is Erdos-Renyi, a 
memoryless policy suffices, i.e., we make the decision based only on the
set of available choices, independent of how long we have spent at the current node. 
Given this, if the cheapest-cost edge is
available right now, clearly it should be chosen. If an optimal policy says to take
the $k$th cheapest edge (among all the potential choices), if it happens
to be the best available, then it follows that we should also take the
$j$th cheapest edge, for $1 \le j \le k$, if it happens to be available. Therefore
the only thing to determine then is the best value $k$, i.e., the one
leading to the policy that minimizes $METT$ from this node.
\end{proof}






\begin{theorem}
Algorithm \ref{times} correctly computes the $METT$ values for the {\sf SoA} model.
\end{theorem}
\begin{proof}
We prove by induction on the nodes removed from $Q$. The expected traversal time of 0 for $t$ is correct by definition. Moreover, by the proof of Corollary \ref{nn}, the expected traversal time of $t$'s nearest neighbor $u$ is $1/p + 1$. 

Suppose there is at least one node whose computed $METT$ value is incorrect, i.e., larger than optimal\footnote{Note that the $METT$ computation always corresponds to a collection of paths from $u$ to $t$; a traversal strategy from $u$ to $t$ restricted to such a path collection can only have expected cost greater or equal to the optimal.}. Among such nodes, let $u$ be one whose {\em true} $METT$ value is minimum. Note that if $v_i$ is removed before $v_j$, then $METT[v_i] \le METT[v_j]$. This follows from the fact that we remove nodes by performing extract-min operations, and that the function $f(p,\cdot)$ is non-decreasing. There must be at least one path from $u$ to $t$, i.e., $u$ must have at least one neighbor whose true optimal expected time to $t$ is strictly less than $u$'s. In fact, $u$ may have several such neighbors. Call them $v_1,...,v_\ell$. If the true $METT$ values of $v_1,...,v_\ell$ are all smaller than $u$'s. then by the induction assumption their $METT$ values are correct, and hence so is $u$'s.

Now suppose some such $v_i$ has not yet been removed. This implies that its computed $METT$ value will be at least $u$'s, even though $v_i$'s true $METT$ value is smaller than $u$'s, which contradicts the induction hypothesis.
\end{proof}

\eat{
The algorithm can be generalized to the following setting: let each edge $e$ have a positive weight $d_e$ indicating the time taken to traverse it, in units of the time steps over which births and deaths occur\footnote{We can also give each edge its own probability $p_e$, but we omit this generalization here for simplicity.}. For practicality, we assume that if we begin traversing $e$, the traversal completes in time $d_e$, even if $e$ dies and/or is reborn during the traversal. To do so, in the $METT$ computation, replace references to $METT[u]$ with $1/p_e + d_e + METT[u]$, where $e$ is the edge being considered for reaching $v$. Similarly, in the routing algorithm, we move from $u$ along an edge $e$ only if $d_e + METT[v]  < METT[v] + 1$.
}
We can now redefine the {\sf CuT} model as the one in which all edge weights are 0. The effect of this is that $v$ neighbor set $N(v)$ is replaced in the routing algorithm with the set of all nodes reachable from $v$, since upon arrival at $u$, we can consider {\em cutting through instantly} to any node that 1) would be an improvement over $u$ and 2) to which there currently is an accessible path from $u$.

\begin{corollary}
Modified appropriately, Algorithm \ref{times} correctly computes the optimal traversal times for the {\sf CuT} model, as well as for the nonnegative integer-weighted model subsuming {\sf SoA} and {\sf CuT}.
\end{corollary}
\begin{proof}
An oracle to compute the probability that there exists an edge between $v$ and $u$, for all pairs $(v,u)$ can be computed in polynomial time by dynamic programming. We omit the details due to lack of space.
\end{proof}


\section{Discussion and Future Work}
\label{sec:discuss}

This paper marks the first step toward a research program aimed at developing a theory of temporal graphs from both stochastic and deterministic (or classical) points of view. We plan to develop the research program in multiple directions. First, the probability distribution results in Sec. \ref{sec:latency} need to be extended beyond simple scenarios such as dynamic random path, especially to scenarios where there are multiple possible (intermittent) paths between the source and the destination. Second, in addition to the T-* properties discussed in Sec. \ref{sec:stgsmg} other properties such as chromatic number, independent set, and dominating set are worth investigating. One interesting question is whether $m$-smashing can be improved by a non-uniform choice of $m$. If the deterministic sequence of graphs is known, then this is akin to a {\em compression} problem where more graphlets will be smashed around times when the temporal ordering does not matter much, and less graphlets will be smashed around other times, thus preserving the temporal structure. However, for a given dynamic random graph model, it may be interesting to develop rules of thumb for non-uniform smashing. In addition to graphlet union (or SmG), graphlet intersection can be interesting since it can be used to quantify redundancy in spatio-temporal paths in a TGS.


\bibliographystyle{plain}
\bibliography{temporalgraph,shortpath}


\end{document}